\documentclass[11pt,a4paper]{article}

\usepackage[utf8]{inputenc}
\usepackage{amsmath, amssymb, amsthm}
\usepackage{physics}
\usepackage{hyperref}
\usepackage{authblk}
\usepackage[colorinlistoftodos]{todonotes}
\usepackage[titletoc,title]{appendix}
\usepackage{comment}

\hypersetup{
    colorlinks,    
    linkcolor={red!50!black},
    citecolor={},
    urlcolor={blue!80!black}
}

\newtheoremstyle{break}
  {12pt}
  {12pt}
  {\slshape}
  {}
  {\bfseries}
  {}
  {\newline}
  {}

\title{\textbf{Hamiltonian engineering via pulses: beyond group averaging}}
\author[1]{Ivan Beschastnyi}
\author[1]{Lucah Patel}
\author[1]{David Tinoco}
\affil[1]{INRIA, MCTAO, France}
\date{\today}

\newtheorem{proposition}{Proposition}[section]

\usepackage[english]{babel}

\usepackage[letterpaper,top=2cm,bottom=2cm,left=3cm,right=3cm,marginparwidth=1.75cm]{geometry}

\usepackage{amsmath}
\usepackage{amsthm}
\usepackage{amssymb}
\usepackage{graphicx}
\usepackage{mathtools}
\usepackage{tikz-cd}
\usepackage{braket}

\newtheorem{theorem}{Theorem}[section]
\newtheorem{lemma}[theorem]{Lemma}

\theoremstyle{definition}
\newtheorem{definition}[theorem]{Definition}

\theoremstyle{remark}
\newtheorem{remark}[theorem]{Remark}

\numberwithin{equation}{section}

  \DeclareMathOperator{\Ad}{Ad}

 \newcommand{\bbar}[1]{\setbox0=\hbox{$#1$}\dimen0=.2\ht0 \kern\dimen0 \overline{\kern-\dimen0 #1}}

 \DeclareMathOperator{\End}{\ensuremath{\mathcal{E}\kern-.125em\mathpzc{nd}}}

 \DeclareMathOperator{\Proj}{\mathcal{P}\kern-.125em\mathpzc{roj}}

 \renewcommand{\setminus}{\smallsetminus}

 \newcommand{\udot}{\ensuremath{{\lower .183333em \hbox{\LARGE \kern -.05em$\cdot$}}}}

  \DeclareMathOperator{\Aff}{Aff}
  \DeclareMathOperator{\Lie}{Lie}

 \newcommand{\cH}{\mathcal{H}}

\newcommand{\cR}{\mathcal{R}}

\newcommand{\C}{\mathbb{C}}

 \newcommand{\R}{\mathbb{R}}

  \newcommand{\Z}{\mathbb{Z}}

 \newcommand{\fg}{\mathfrak{g}}
  
  \newcommand{\fh}{\mathfrak{h}}

 \newcommand{\p}{\partial}

  \newcommand{\fsu}{\mathfrak{su}}
  \newcommand{\fsl}{\mathfrak{sl}}

 \DeclareMathOperator{\sign}{sign}
 
 \DeclareMathOperator{\conv}{conv}
 \DeclareMathOperator{\spann}{span}

\begin{document}

\maketitle

\begin{abstract}
We develop a geometric framework for Hamiltonian engineering in finite-dimensional quantum systems using ideal control pulses. Starting from a bilinear Schrödinger equation with unbounded control amplitudes, we construct the closed pulse group and use extensions of control systems and Filippov's relaxation theorem to obtain a family of effective Hamiltonians given by the convex hull of the drift's adjoint orbit plus the Lie algebra of the pulse group. The geometry of this orbitope describes possibilities beyond group averaging. Using the isotypic decomposition of the adjoint representation, we characterize its affine hull and show that the group average lies in its interior. This yields locally accessible families of effective Hamiltonians around the invariant part of the drift. We apply the framework to recover the necessary and sufficient condition for dynamical decoupling from arbitrary interactions with a finite-dimensional bath. For connected abelian pulse groups, we describe the relevant representation decomposition through restricted roots. Finally, for qubit networks with identical pairwise couplings, we give a qualitative characterization and quantitative estimation of effective Hamiltonians that can be generated using our strategy.
\end{abstract}

\section{Introduction}
\label{sec:intro}

Many practical quantum systems can be described through finite-dimensional control systems of the form
\begin{equation}
\label{eq:unitary_control_system}
    i\dot U = (H_0 +\sum_{j=1}^m u_j H_j )U, \qquad U(0) = I.
\end{equation}
Here $U\in SU(n)$ is a unitary operator, $H_0,\dots,H_m$ are traceless Hermitian matrices and $(u_1,\dots,u_m)\in\R^m$ are controls. Note that we assume the control functions can take arbitrarily large values. By applying very large controls over very small intervals of time we achieve in the limit pulses -- instantaneous unitary transformations. No ideal pulses, of course, exist, but depending on the initial Hamiltonians and real control limitations, the practical non-ideal pulses can sometimes be very close to the ideal ones. Therefore, it makes sense to study what can and cannot be achieved by using pulses in combination with the more standard classes of control strategies.

Pulse strategies have been employed in many situations, in particular in dynamical decoupling, which is used to suppress unwanted interactions between two coupled quantum systems~\cite{ViolaKnillLloyd1999}. Examples include separating a system from the surrounding bath or disconnecting two connected qubits via a pure control strategy, without additional engineering solutions. Such applications often modify the drift Hamiltonian $H_0$ by averaging out certain terms. However, as we will see, one can use pulses to do much more and generate a wide class of effective Hamiltonians depending on the drift and pulses available. The main question of this paper is:
\begin{center}
    \textit{What Hamiltonians can be generated by combining bounded controls and pulses?}
\end{center}
Using standard techniques from geometric control theory and a little bit of representation theory we give basically a complete answer to this question. Our main tool is the technique of extensions of control systems. In the case of right-invariant problems on Lie groups, this technique was developed in~\cite{JurdjevicKupka1981Accessibility,JurdjevicKupka1981Semisimple}, and has also been used in many other concrete situations~\cite{AgrachevBaryshnikovSarychev2016}. Its advantage is that it allows us very quickly to cover all the possibilities pulse strategies have to offer and reduce the problem to a representation-theoretic one, where algebraic methods can be employed efficiently. This is in contrast with many previous papers that relied on Magnus expansions for similar tasks~\cite{HaeberlenWaugh1968,ViolaKnillLloyd1999}, which arguably require more computational work.

Another difference is that in our case the set of all pulses will form a continuous group, whereas many works in dynamical decoupling use finite groups or discrete sets~\cite{ViolaKnillLloyd1999,ReadEtAl2025Platonic,ReadEtAl2025Factorization,TripathiEtAl2025,NguyenEtAl2026}. In the Hamiltonian simulation literature~\cite{BennettEtAl2002,WocjanEtAl2002}, the authors of~\cite{WocjanEtAl2002} develop local-control constructions for simulating pair interactions, including decoupling and time reversal, and derive bounds on simulation overhead for a fixed finite family of pulses. Other recent examples include sequences based on finite Platonic rotation groups~\cite{ReadEtAl2025Platonic} and their factorization into subgroups to exploit existing Hamiltonian symmetries~\cite{ReadEtAl2025Factorization} as well as robust constructions~\cite{ChoiEtAl2020}. Finite group constructions were also used experimentally demonstrated qudit decoupling using the Heisenberg--Weyl group~\cite{TripathiEtAl2025} and selective decoupling of local interactions using subgroups of the projective Pauli group and classical additive codes~\cite{NguyenEtAl2026}. We could, in theory, restrict our analysis to discrete families of pulses as well, but 
this would exclude certain strategies that seem to be perfectly realizable in practice and lead to an incomplete answer to the main question of the paper. In addition, the standard averaging of the drift Hamiltonian over a group action results only in a small subset of possible Hamiltonians that can be engineered. The ability to generate a large class of Hamiltonians is important, among other things, in quantum simulation, where one of the main tools is Trotterization~\cite{Lloyd1996}. It is used to simulate a complex Hamiltonian by switching between simpler ones usually realizable on a quantum computer. But it might very well be possible to generate quite complex Hamiltonians through a suitable pulse strategy, without splitting it into smaller Hamiltonians first. This can reduce the time and complexity of the corresponding quantum algorithms by reducing the depth of quantum circuits.

Closer to our approach is the work~\cite{KhanejaBrockettGlaser2001}, where the authors formulate pulse design as time-optimal control on compact Lie groups and relate fast control to motion on homogeneous spaces. Their analysis connects the time needed to synthesize a propagator with adjoint-orbit geometry. One could use our method as a step towards proving controllability results, even though this is not the goal of the paper. The Lie algebraic foundations of finite-dimensional quantum control, including controllability criteria and decompositions of quantum dynamics, are presented systematically in~\cite{DAlessandro2021}. Representation-theoretic methods have been used for proving general controllability (see, for example,~\cite{DAlessandro2024}).

It should be noted that the goal of this work is to explore which Hamiltonians can be engineered theoretically, setting aside practical difficulties associated with non-ideal pulses and the choice of appropriate time scales. We give references for concrete realizations of individual steps of our construction throughout the text. 

Let us give a bird's-eye view of Hamiltonian engineering for the system~\eqref{eq:unitary_control_system} and the structure of the paper. In this work we pull back the Lie algebra structure from the Lie algebra of traceless skew-Hermitian matrices $\fsu(n)$ to the space of traceless Hermitian matrices of order $n$. Indeed, if $H$ is a traceless Hermitian matrix, then $-iH$ is an element of $\fsu(n)$ and this map is clearly bijective. Taking the pull-back of the usual matrix Lie bracket gives us an equivalent Lie bracket on the space of Hermitian matrices:
\begin{equation}
    \label{eq:Lie_bracket}
    [H_1,H_2]:= \frac{1}{i}(H_1H_2 - H_2H_1).
\end{equation}
\textit{For this reason, whenever we write $\fsu(n)$, we mean the space of traceless Hermitian matrices of order $n$ endowed with the bracket above.}

We note that the system~\eqref{eq:unitary_control_system} can be rewritten as a differential inclusion
\begin{equation}
    \label{eq:diff_inclusion}
    i\dot U \in \cH U, \qquad U(0) = I,
\end{equation}
where 
$$
\cH =\{H_0 + u_1H_1 +\dots + u_mH_m: (u_1,\dots,u_m)\in \R^m \}\subset \fsu(n).
$$
\textit{Differential inclusions} are nothing but set-valued differential equations, meaning that an absolutely continuous curve $U:[0,T]\to SU(n)$ is a solution to~\eqref{eq:diff_inclusion} if~\eqref{eq:diff_inclusion} is satisfied at almost every time. \textit{The reachable set} of~\eqref{eq:diff_inclusion} is the set of all points $U_T \in SU(n)$ such that there exists a solution $U:[0,T]\to SU(n)$, for which $U(0)= I$ and $U(T)=U_T$. 

To engineer Hamiltonians using pulses, we use the technique of extensions of control systems. The idea is that we can replace $\cH$ with $\tilde \cH \supset \cH$, such that the closure of the reachable set does not change. To do so we first define the pulse group (it is an extension of the notion of control group in the dynamical decoupling literature~\cite{ViolaKnillLloyd1999}).
\begin{definition}
\label{def:pulse_group}
The pulse group $G$ is the smallest closed Lie subgroup of $SU(n)$ whose Lie algebra $\fg$ contains $\Lie\{H_1,\dots,H_m\} $, the Lie algebra generated by the Hamiltonians $H_1,\dots,H_m$.
\end{definition}
Every pulse in $G$ can be approximated arbitrarily well by a sequence of pulses of the form $e^{-i \theta H_i}$, $\theta \in \R$. Essentially $G$ is the theoretically largest set of pulses.

Next we combine pulses from $G$ with evolution under the drift Hamiltonian $H_0$. We will see in Section~\ref{sec:pulse_group} that it results in a total extension of the form:
$$
\tilde \cH = \conv(Ad_G H_0 + \fg), \qquad Ad_G H_0 := \{Ad_g H_0 = gH_0g^\dagger: g \in G\}.
$$
This gives a large set of time-independent Hamiltonians that can be engineered through pulses and constant controls. 

\begin{definition}
    We call the set $\conv(Ad_G H_0)$ the \textit{orbitope of effective Hamiltonians}.
\end{definition}
We will see that this set $\tilde \cH$ is basically determined by this orbitope. Recall that an orbitope is the convexification of an orbit of a representation of a compact group~\cite{orbitopes}. Our main task reduces to the description of this set, which can be achieved through a little bit of representation theory.

The map 
$$
g\in G \mapsto Ad_g \in End(\fsu(n))
$$
is a representation of the group $G$. Since the action of $G$ is orthogonal with respect to the scalar product
\begin{equation}
\label{eq:scalar_product}
    \langle H_1,H_2 \rangle = \Tr(H_1H_2), \qquad \forall H_1,H_2 \in \fsu(n), 
\end{equation}
we can decompose $\fsu(n)$ into irreducible representations of $G$, which we collect into the \textit{isotypic components}:
\begin{equation}
    \fsu(n) = V_0 \oplus \bigoplus_\lambda V_\lambda,
\end{equation}
where $V_0$ is the subspace of vectors fixed by the action of $Ad_G$, $\lambda$ belongs to the set of non-trivial irreducible representations of $G$ and $V_\lambda \simeq W_\lambda^{m_\lambda}$ are the isotypic components, which are subrepresentations of $Ad_G$. Here $W_\lambda$ is a model for the irreducible representation $\lambda$. 

In Section~\ref{sec:orbitopes}  we will see that the set $\tilde \cH$ can be partially described by the space generated by $\conv Ad_G H_0$, which can be described effectively via projections to the isotypic components. Two projections will play a very important role:
\begin{align*}
    &\pi_0:\fsu(n)\to V_0,\\
    &\pi_0^\perp:\fsu(n)\to V_0^\perp,
\end{align*}
namely the projection $\pi_0$ to the subspace of fixed elements and the projection $\pi_0^\perp$ to its orthogonal complement. In particular, we will see that $\conv Ad_G H_0$ lies inside a minimal affine subspace $\Aff(Ad_G H_0)\subset \fsu(n)$ passing through $\pi_0 H_0$. This means that every Hamiltonian inside a sufficiently small ball centered at $\pi_0 H_0$ and contained inside $\Aff(Ad_G H_0)$ can be engineered. This is of particular interest when $\pi_0 H_0 = 0$, since the differential inclusion given by a small ball centered around zero will be equivalent to a driftless control system with plenty of Hamiltonians to choose from.

Finally, in Section~\ref{sec:applications}, we give some applications. For example, we recover the classical dynamical decoupling condition.
\begin{theorem}
    \label{thm:decoupling}
    Consider a bipartite finite-dimensional quantum system consisting of a controlled system of interest and an uncontrollable bath, with total Hamiltonian of the form
\begin{equation}
    \label{eq:bi-hamiltonians}
    H(u) = H_S(u) \otimes I + I \otimes H_B + H_{int}
\end{equation}
    where $H_S(u) = H_0 + u_1 H_1 + \dots + u_m H_m$ is the Hamiltonian of a controlled system, $H_B$ is the Hamiltonian of the bath and $H_{int}$ is the interaction Hamiltonian, which is assumed to be unknown. 
    
    Then the controlled system can be decoupled from any bath for any possible $H_{int}$ by a single control strategy if and only if the subspace of fixed elements under the action of $G$ is trivial, i.e. $V_0 = \{0\}$. The universal decoupling strategy will result in effective Hamiltonians of the form
    $$
    H_{eff} = H_{pulse}\otimes I + I \otimes H_B,
    $$
    where $H_{pulse}\in \fg$.
    \end{theorem}
This theorem implies that any decoupling strategy that does not assume any a priori knowledge about the environment interaction will result in a trivial Hamiltonian acting on the system of interest (if it annihilates any possible $H_{int}$, it must annihilate the whole of $H_S(u)$ as well). In particular, the only achievable states of the controlled system are the elements of the control group $G$.

Another application is given by the control of a network of identical qubits with identical couplings. Assume that we have $N$ two-level systems which are coupled according to a graph $\Gamma$. This means that vertices represent qubits and edges the couplings between them. Assume that we have a controlled Hamiltonian of the form
\begin{equation}
    \label{eq:qubit_hamiltonian}
    H_0 + \sum_{j=1}^{2N}u_jH_{j}
\end{equation}
where
\begin{equation}
    \label{eq:qubit_drift}
    H_0 = \sum_{j=1}^N \omega Z_j + \sum_{j \sim k} \beta X_jX_k + \gamma Y_jY_k + \delta Z_jZ_k,
\end{equation}
\begin{equation}\label{eq:control_hamiltonians}
        H_j = \begin{cases}
        X_j, & \text{if } j = 1,\dots N;\\
        Y_{j-N} & \text{if } j = N+1,\dots 2N.
    \end{cases} 
\end{equation}
Here the summation over $j\sim k$ means summation over all adjacent pairs of vertices, the real constants $\beta, \gamma,\delta$ are not all zero, $X,Y,Z$ denote the Pauli matrices and the index indicates the site on which each matrix acts. We have used the usual shorthand notation $X_j = I_1 \otimes \dots \otimes X_j \otimes \dots \otimes I_N $ and similarly for $Y$ and $Z$. For example, capacitive coupling between superconducting qubits gives rise to transverse interactions~\cite{Krantz2019}. Effective models can also contain residual $ZZ$ interactions~\cite{MagesanGambetta2020}, which a re often seen as parasitic and are represented here by a nonzero coefficient $\delta$.

For this system, an interesting question motivated by quantum computation is: which single- and two-qubit Hamiltonians can be generated? Is it possible to generate Hamiltonians acting on more than two qubits? With our analysis we are able to give a qualitative description of all effective Hamiltonians.
\begin{theorem}
\label{thm:qubit_networks}
    Consider a system with Hamiltonian~\eqref{eq:qubit_hamiltonian}-\eqref{eq:control_hamiltonians}.
    \begin{enumerate}
        \item If $u_j=0$ for $j=N+1,\dots ,2N$ we have
\begin{align}
    \Aff (Ad_G H_0) &  = \beta\sum_{j \sim k} X_jX_k +\omega\left(\bigoplus_{j}\spann \{ Y_j,Z_j\}\right) \oplus\nonumber \\
    &\oplus \frac{(\delta - \gamma)}{2}\left(\bigoplus_{j\sim k}\spann \{Y_jY_k-Z_jZ_k, Y_jZ_k+Z_jY_k\}\right) \oplus \nonumber\\
    &\oplus  \frac{(\delta + \gamma)}{2}\left(\bigoplus_{j\sim k}\spann \{Y_jY_k+Z_jZ_k, Y_jZ_k-Z_jY_k\}\right).\label{eq:orbitope_many_qubits_X_pulses}
\end{align}
        \item Otherwise,
        $$
        \Aff (Ad_G H_0) = \omega\bigoplus_{j\in\{1,\dots,N\}}(\fsu(2)_j \otimes I_{\{1,\dots,N\}\setminus \{j\}}) \oplus \bigoplus_{j\sim k}(\fsu(2)_j\otimes \fsu(2)_k\otimes I_{\{1,\dots,N\}\setminus \{j,k\}}).
        $$
    \end{enumerate}

\end{theorem}
 What does this mean in practical terms? If we only allow for rotations around $X$ as pulses, then no $XX$ coupling can be killed and as a result no single- or two-qubit Hamiltonians are in the space of effective Hamiltonians (unless there are only one or two qubits altogether). If $\beta = 0$, then we can generate a variety of single-qubit and nearest-neighbor two-qubit Hamiltonians, but not all of them. If instead we also allow for $Y$-pulses, then all single-qubit Hamiltonians can be generated when $\omega \neq 0$ as well as two-qubit Hamiltonians for neighboring qubits. However, also in this case no long-distance qubit pair Hamiltonians can be generated. Note that this theorem gives a description of the space spanned by $\conv Ad_G H_0$ and not the convexification itself. Therefore, this description is enough to know which effective Hamiltonians can be achieved modulo a multiplicative constant. Note that we also have the Lie algebra of the pulse group $G$, which can be used to generate exponentials of elements in $\fg$.

    \begin{remark}
        It should be emphasized that even though we can only generate nearest-neighbor qubit Hamiltonians using a combination of pulses and bounded controls, it does not mean that long-distance qubit gates cannot be achieved. The Hamiltonian~\eqref{eq:qubit_hamiltonian} under generic assumptions is controllable, which implies that there exists a control (even without pulses) that would allow us to realize a gate acting on an arbitrary pair of qubits. It will simply not be possible to realize such a gate as a single exponential of one of the effective Hamiltonians.
    \end{remark}

Theorem~\ref{thm:qubit_networks} gives a qualitative characterization of possible effective Hamiltonians. This means that, for example, in the first case we can generate all Hamiltonians of the form
\begin{align}
    &\beta\sum_{j \sim k} X_jX_k +\omega\sum_{j}(a_{j} Z_j +b_j Y_j) +\nonumber \\
    + &\frac{(\delta - \gamma)}{2}\sum_{j\sim k}\left(a_{jk}(Z_jZ_k-Y_jY_k)+ b_{jk}(Y_jZ_k+Z_jY_k)\right) + \nonumber \\
    + &\frac{(\delta + \gamma)}{2}\sum_{j\sim k}\left(\tilde a_{jk}(Z_jZ_k+Y_jY_k)+ \tilde b_{jk}(Y_jZ_k-Z_jY_k)\right).\label{eq:eff_ham_qubit_network_X}
\end{align}
with $a_j,b_j,a_{jk},b_{jk},\tilde a_{jk},\tilde b_{jk}$ being real parameters. However, Theorem~\ref{thm:qubit_networks} does not say anything about the \textit{range} of those coefficients. Similarly, the second part of Theorem~\ref{thm:qubit_networks} says that we can engineer all of the Hamiltonians of the form
\begin{align}
    \sum_{j=1}^N&(a_{1,j} X_j +a_{2,j} Y_j+a_{3,j}Z_j) 
    + \sum_{j\sim k}\left(\sum_{p,q=1}^3a_{pq,jk}P_{p,j}P_{q,k}\right), 
    \label{eq:eff_ham_qubit_network_XY}
\end{align}
    where $(P_{1,j},P_{2,j},P_{3,j})=(X_j,Y_j,Z_j)$ and all the coefficients are assumed to be real. Here arbitrary single-qubit Hamiltonians are supplied by $\fg$, independently of $\omega$. Thus this family is considered in $\tilde\cH=\conv(\Ad_G H_0)+\fg$.

To obtain a quantitative estimate we need to find a set $K\subset \conv (Ad_G H_0)$ that can be described by a finite number of inequalities. This would allow us to quickly check if a given combination of coefficients gives rise to a realizable Hamiltonian, even though this will give only a sufficient condition for realizability. Complementary to Theorem~\ref{thm:qubit_networks} we have the following result.
\begin{theorem}
\label{thm:quantitative_qubit_networks}
    Under the assumptions of Theorem~\ref{thm:qubit_networks},
    \begin{enumerate}
        \item For the first case, if the following inequality is satisfied:
        $$
        \sum_{j=1}^N \sqrt{a_j^2+b_j^2} + \sum_{j\sim k}\left( \sqrt{a_{jk}^2+b_{jk}^2} + \sqrt{\tilde a_{jk}^2+\tilde b_{jk}^2}\right) \leq 1
        $$
        then the Hamiltonians of the form~\eqref{eq:eff_ham_qubit_network_X} can be engineered.
        \item For the second case, if the coefficients $a_{pq,jk}$ for $p,q=1,2,3$ and $j\sim k$ satisfy the following inequality:
        $$
        \frac{1}{R}\sum_{j\sim k} \sqrt{\sum_{p,q=1}^3 a_{pq,jk}^2 }\leq 1,
        $$
        where
        $$
        R = \min\left\{\max\{|\beta|,|\gamma|,|\delta|\},\frac{1}{\sqrt{3}}\max\left\{|\beta|+|\gamma| - |\delta|,|\gamma|+|\delta| - |\beta|,|\delta|+|\beta| - |\gamma| \right\}\right\},
        $$
        then the Hamiltonians of the form~\eqref{eq:eff_ham_qubit_network_XY} can be engineered (no restrictions on $a_{1,j}, a_{2,j}, a_{3,j}$). 
    \end{enumerate}
\end{theorem}
    
    We give a proof of this theorem in the appendix together with proofs of some lemmas related to the structure of orbitopes stated in Section~\ref{sec:orbitopes}.

\section{Pulse group and the orbitope of effective Hamiltonians}

\label{sec:pulse_group}

Consider the controlled bilinear Schrödinger equation~\eqref{eq:unitary_control_system}. As discussed previously, we assume that $u_j\in \R$ can take unbounded values. Then, as we take constant controls $u/\varepsilon$ and the final time $T=\varepsilon$ in the limit $\varepsilon\to 0$ the final point of the solution of~\eqref{eq:unitary_control_system} approaches
a solution of 
$$
    i\dot U = \left(\sum_{j=1}^m u_j H_j\right) U
$$
at time $1$. Note that since $T\to 0$, we can essentially realize instantaneous jumps (pulses) given by multiplications by unitary operators of the form
\begin{equation}
    \label{eq:exponentials}
    e^{-it(\sum_{j=1}^m u_j H_j)}.
\end{equation}
We can now compose individual pulses. As a result, we can achieve instantaneous jumps given by multiplication by elements of the group generated by the exponentials~\eqref{eq:exponentials}. As can be seen, for example, by the Baker-Campbell-Hausdorff formula, this group can be characterized as the unique subgroup of $SU(n)$ that integrates the Lie algebra generated by $\{H_1,\dots,H_m\}$. Note that this Lie subgroup is not necessarily closed, so we need to take the closure, which is exactly the pulse group $G$ defined in Definition~\ref{def:pulse_group}.

\begin{remark}
We do not discuss here concrete realizations of elements of pulse groups. For an example of constructions of Lie-bracket extensions using fast-oscillating controls, see~\cite[Sections~6.3--6.4]{AgrachevBaryshnikovSarychev2016}.
\end{remark}

\begin{remark}
\label{rmk:non-equailty}
    The Lie algebra of $G$ contains $\Lie\{H_1,\dots,H_m\}$, but is not necessarily equal to it. An example is a single Hamiltonian $H_1$, which generates a dense winding inside the maximal torus $T$ of $SU(n)$. It is clear that the Lie algebra $\mathfrak t$ of $T$ has dimension $n-1\geq 1$.
\end{remark}

Now that we have covered all possible pulses, let us mix pulses with ordinary controls. More precisely, let $g\in G$ be a pulse. Then we can realize
$$
ge^{-it H_0}g^\dagger = e^{-it Ad_g H_0},
$$
where $Ad_g H_0 = gH_0 g^\dagger$. Essentially, by applying a pulse-drift-pulse strategy, we were able to generate the propagator along $Ad_g H_0$. Combining all possible pulses, we see that we can generate propagators along Hamiltonians $Ad_G H_0$.
Gathering everything together, we can generate Hamiltonians of the form
$$
Ad_G H_0 + \fg.
$$

We can further extend the set of admissible velocities using \textit{Filippov's relaxation}, which is a generalization of the Trotter product formula. Before stating it, we state some definitions from the introduction in a mathematically precise way.

\begin{definition}
    Let $\cH \subset \mathfrak{su}(n)$ be a subset of Hermitian matrices. An absolutely continuous curve $U:[0,T]\to SU(n)$ satisfying 
\begin{equation}
\label{eq:diff_incl}
        i\dot U(t) \in \cH U(t)
\end{equation}
    for almost all $t\in[0,T]$ is called admissible. The reachable set from a point $U_0$ of the differential inclusion~\eqref{eq:diff_incl} at a time $t$ is the set
    $$
    \cR_{t}^\cH(U_0) = \{U(t): U:[0,t] \to SU(n), \text{ admissible},U(0)=U_0\}
    $$
    We also define the reachable set of~\eqref{eq:diff_incl} up to time $T$ as
    $$
    \cR_{\leq T}^\cH(U_0) = \bigcup_{t\in [0,T]} \cR_{t}^\cH(U_0).
    $$
\end{definition}

Since all of the differential inclusions in this paper are right-invariant, we can translate the reachable sets using right multiplication:
$$
\cR_{\leq T}^\cH(U_0 U_1) = \cR_{\leq T}^\cH(U_0)U_1, \qquad \forall U_0,U_1\in SU(n).
$$
In particular, it is enough to study reachable sets from the identity, for which we use the shortened notation
$$
\cR_{\leq T}^\cH := \cR_{\leq T}^\cH(I).
$$

We can now state a version of Filippov's theorem that applies to our case.
\begin{theorem}
\label{thm:Filippov}
Let $\cH\subset \mathfrak{su}(n)$ be a subset of Hermitian matrices of order $n$. Then 
$$
\overline{\cR^{\cH}_{\leq T} } = \overline{\cR^{\conv(\cH)}_{\leq T} }.
$$
\end{theorem}

Filippov's theorem says that convexifying the set of admissible velocities does not change the closure of the set of points that we could reach with the original system. 

\begin{remark}
    Usually Filippov's theorem is stated for general differential inclusions. Classical relaxation theorems and their hypotheses are discussed in~\cite{AubinCellina1984}. A constructive proof can be found in~\cite[Theorem 8.7]{Agrachev2004}.
\end{remark}

\begin{remark}
Theorem~\ref{thm:Filippov} is a generalization of the familiar Trotter product formula:
\begin{equation}
\label{eq:trotter}
    \lim_{N\to \infty} (e^{-it H_1/N}e^{-it H_2/N})^N = e^{-i2t \frac{(H_1 + H_2)}{2}}.
\end{equation}
In this case $\cH = \{H_1,H_2\}$ and by applying a control strategy of switching between $H_1$ and $H_2$ on constant intervals of time, we 
approximate the propagator along a Hamiltonian inside $\conv \cH$. Usually the right side of~\eqref{eq:trotter} is written in the form $e^{-it (H_1 + H_2)}$, but in this case there is a small time inconsistency. Inside $e^{-it (H_1 + H_2)}$ it seems that only time $t$ has passed, but on the left-hand side of~\eqref{eq:trotter} we switch between Hamiltonians $H_1$ and $H_2$ for the total time of $2t$. We can generalize this to an arbitrary convex combination of the two Hamiltonians in a straightforward manner:
$$
    \lim_{N\to \infty} (e^{-it \lambda H_1/N}e^{-it (1-\lambda )H_2/N})^N = e^{-it (\lambda H_1 + (1-\lambda)H_2)}.
$$

\end{remark}

\medskip

Let us now apply Theorem~\ref{thm:Filippov} to our case. We collect all the possible Hamiltonians that can be achieved through pulses and convexify this set. This includes the directions that can be achieved by pulses in $G$ and also by the combination of pulses and drifting. This results in
\begin{equation}
    \label{eq:relaxed_system}
    \tilde \cH = \conv(Ad_G H_0 + \fg) = \conv(Ad_G H_0) + \fg.
\end{equation}
Since the adjoint action preserves the bi-invariant scalar product on $\fsu(n)$ given by~\eqref{eq:scalar_product},
and $\fg$ is obviously a subrepresentation of $G$, we can replace $H_0$ in~\eqref{eq:relaxed_system} with the projection of $H_0$ to the orthogonal complement of $\fg$. To simplify the notation below, we make the following assumption, which does not reduce generality.

If we assume $H_0$ is orthogonal to $\fg\subset \fsu(n)$, then we would get
$$
\tilde \cH = \conv(Ad_G H_0) \oplus \fg.
$$
Therefore, we only need to understand the geometry of the convexification of the orbit $Ad_G(\pi_{\fg^\perp}H_0)$, where $\pi_\fg^\perp: \fsu(n)\to \fg^\perp$ is the projection to the orthogonal complement of $\fg$.  

\section{Orbitopes}
\label{sec:orbitopes}

The set $\conv(Ad_G H_0)$ is an example of an orbitope. Here we collect facts about them, which will be useful in the sequel.
\begin{definition}
    Let $G$ be a compact Lie group and $(V,\rho)$ be an orthogonal representation. An orbitope associated to $v\in V$ is the convexification of the corresponding orbit
    $$
    \conv(O_v) := \conv\{\rho(g)v:g\in G\}. 
    $$
\end{definition}

Various properties of orbitopes can be found in~\cite{orbitopes,Kobert2021}. To understand the geometry of the orbitope $\conv(Ad_G H_0)$ we need a little bit of representation theory. Since $G$ is a compact group, all of its irreducible representations are finite-dimensional. We denote the elements of the space of non-trivial irreducible representations using $\lambda$. The Lie algebra $\fsu(n)$ as a representation of $G$ can be decomposed into a sum of irreducible subrepresentations. This decomposition in general is not unique, because there might appear several subrepresentations that are isomorphic to a given irreducible representation $W_\lambda$. However, if we collect together all spaces isomorphic to a given one, we obtain the isotypic component $V_\lambda \simeq W_\lambda^{m_\lambda}$ and the decomposition into isotypic components
$$
\fsu(n) = V_0 \oplus (\bigoplus_\lambda V_\lambda)
$$
is unique. In the formula above $V_0$ is the subspace of fixed elements under the representation $Ad_G$. We can then decompose $H_0$ corresponding to projections onto isotypic components. Let $\pi_0:\fsu(n)\to V_0$ and $\pi_0^\perp : \fsu(n)\to V_0^\perp$ be orthogonal projections.

It is clear that
$$
\conv (Ad_G H_0) = \conv(Ad_G (\pi_0 H_0) + Ad_G (\pi^\perp_0 H_0)) = \pi_0 H_0 + \conv(Ad_G (\pi^\perp_0 H_0))
$$
Therefore, the question reduces to investigating the orbitope $\conv(Ad_G (\pi^\perp_0 H_0))$. In order to shorten the formulas we assume in this section without any loss of generality that $V_0 = \{0\}$. In this case, $\pi^\perp_0 H_0 = H_0$.

We have the following useful lemma.
\begin{lemma}
\label{lem:0_inside}
Let $G$ be a compact group and $(V,\rho)$ a representation of $G$ on $V$ that has no nonzero fixed elements. Let $v\neq 0$ be an element of $V$. Then $0$ belongs to the interior of the orbitope $\conv(O_v)\subset V$.
\end{lemma}
For completeness we give the proof in the appendix. As a result, under the assumptions of Lemma~\ref{lem:0_inside}, the orbitope $\conv(Ad_G H_0)$ will contain Hamiltonians inside a small ball contained in $\spann (Ad_G H_0)$ and centered at $\pi_0 H_0$. This subspace of $\fsu(n)$ can be further described, which leads to the following lemma.
\begin{lemma}[{\cite[Section~2]{orbitopes}}]
\label{lemm:decomposition}
    Let $G$ be a compact group and $(V,\rho)$ a representation of $G$ on $V$ that has no nonzero fixed elements and consider its decomposition into isotypic components:
    $$
    V= \bigoplus_{\lambda}V_\lambda.
    $$
    Denote by $\pi_\lambda: V \to V_\lambda$ the orthogonal projections. Then
    $$
    \spann( O_v) = \bigoplus_\lambda \spann(O_{\pi_\lambda v}).
    $$
\end{lemma}
This lemma implies that we can compute separately $\spann Ad_G (\pi_\lambda H_0) $ for each individual isotypic component and join them together using direct sums to determine the plane that the orbitope generates. Then, from Lemma~\ref{lem:0_inside}, we can deduce that a sufficiently small ball in $\spann Ad_G (\pi_\lambda H_0) $ is contained inside $\conv Ad_G (\pi_\lambda H_0) $.

What remains is to show how to calculate each individual $\spann Ad_G (\pi_\lambda H_0) $. This will largely depend on the relative position of $\pi_\lambda H_0$ inside the isotypic component $V_\lambda$ and the action of $G$. In general, a complete description of such orbitopes is a difficult problem (see~\cite{orbitopes,Barvinok2005,Kobert2021} for some known examples). To partially solve this problem in the examples of this article, we will look for a convex body $K\subset \spann Ad_G (\pi_\lambda H_0)$ that can be described via a finite number of inequalities.


    


\section{Applications}
\label{sec:applications}

\subsection{Dynamical decoupling}
\label{sec:dynamical_decoupling}

In this section we prove Theorem~\ref{thm:decoupling}. From the Hamiltonian~\eqref{eq:bi-hamiltonians} we see that the pulse group is isomorphic to the pulse group $G$ generated only by $\{H_1,\dots,H_m\}$. Since the controls only act on the system and not on the bath, we have that the orbitope of effective Hamiltonians is determined by the representation $\rho(g) = Ad_g \otimes I$, which can be seen as a tensor product of the adjoint representation and the trivial representation. 

Denote by $\mathfrak H_S$ and $\mathfrak H_B$ the Hilbert spaces of the system and of the bath. Note that $\fsu(\mathfrak H_S)\otimes I$, $I\otimes \fsu(\mathfrak H_B)$ and $\fsu(\mathfrak H_S)\otimes \fsu(\mathfrak H_B)$ form subrepresentations. Therefore, to prove the theorem it is enough to find necessary and sufficient conditions for the orbitope
$$
\conv\{\rho(g)H_{int}\}
$$
to contain zero.

We can decompose $\fsu(\mathfrak H_S)$ into irreducible representations of $G$
$$
\fsu(\mathfrak H_S) = V_0 \oplus (\bigoplus_\lambda V_\lambda)
$$
Note that the subspace of trivial action for $\rho$ on $\fsu(\mathfrak H_S)\otimes \fsu(\mathfrak H_B)$ is given by
$$
V_0 \otimes \fsu(\mathfrak H_B).
$$
Therefore, since $H_{int}\in \fsu(\mathfrak H_S)\otimes \fsu(\mathfrak H_B)$ a necessary condition for decoupling from any bath coupling is $V_0 = \{0\}$. This is also a sufficient condition by Lemma~\ref{lem:0_inside}.

Finally we note that if we use a strategy that cancels all possible interaction terms in $\fsu(\mathfrak H_S)\otimes \fsu(\mathfrak H_B)$, then it must cancel out all possible terms in $\fsu(\mathfrak H_S)\otimes I$. Indeed, choose a basis $s_i$ of $\fsu(\mathfrak H_S)$. Then the interaction term can be written as
$$
H_{int}=\sum_{i,\alpha}a_{i\alpha} s_i\otimes b_\alpha,
$$
where all $b_\alpha$ are independent. Note that each individual term in the sum lies in a separate subrepresentation of $G$. Since $G$ acts trivially on the bath term, it will act with a single $g\in G$ on all of $s_i$ simultaneously. Then any control sequence that makes the interaction term vanish must annihilate each $s_i$. But it is a basis of $\fsu(\mathfrak H_S)$, and therefore also $H_S$ can be decomposed in this basis. As a consequence, $H_S$ will also vanish. This only leaves us with $I \otimes H_B$, which is a fixed element of the $\rho$-action and cannot be canceled out, and elements that generate pulses $\fg \otimes I$.

\subsection{The case of a continuous abelian pulse group}
\label{sec:abelian_groups}

Before proving Theorem~\ref{thm:qubit_networks}, let us make some general remarks about when $G$ is abelian and continuous. This will be the case in the first part of Theorem~\ref{thm:qubit_networks} and in models with $ZX$ cross-resonance in controlled Hamiltonians~\cite{MagesanGambetta2020}.

Assume that $G$ is abelian, which means that all $H_1,\dots,H_m$ commute. In this case $\conv (Ad_G H_0)$ can be characterized quite effectively, due to the fact that $\mathfrak{su}(n)$ is a real simple Lie algebra. Consider a maximal abelian subalgebra (also called a Cartan subalgebra) $\fh$ which extends $\fg$. Recall that $\dim \fh = n-1$. Then complexify everything. The complexification of $\fsu(n)$ is the Lie algebra $\mathfrak{sl}(n,\C)$. As a result, we can decompose $\mathfrak{sl}(n,\C)$ into irreducible representations with respect to the restriction of the adjoint action of $\fh_\C$:
$$
ad_X Y := [X,Y], \qquad \forall X \in \fh, Y\in \mathfrak{sl}(n,\C).
$$
This gives a decomposition of the form
$$
\fsl(n,\C) = \fh_\C \oplus (\bigoplus_{\alpha \in \Delta} \fg_\alpha) \oplus (\bigoplus_{\alpha \in \Delta} \fg_{-\alpha}),
$$
where $\Delta$ is called the \textit{set of positive roots}, which is a discrete subset of $\fh^*$ satisfying certain algebraic properties and $\fg_\alpha$ are called \textit{root spaces}. One can prove that $\dim \fg_\alpha = 1$. Moreover, it is possible to find
a set of generators $Q_\alpha$ of $\fg_\alpha$, with the property that $Q_\alpha^\dagger = Q_{-\alpha}$ and such that
$$
[H,Q_\alpha] = -i\alpha(H)Q_\alpha, 
$$
where, by an abuse of the previous notation, $H \in \fh_\C$. Usually there is no $-i$ term in the formula above, because the standard matrix commutator is used and not the commutator~\eqref{eq:Lie_bracket} that we adopt in this paper. The remaining commutators can be found in~\cite{Humphreys1972} (modulo the multiplicative $-i$ terms), even though they are not important for our purposes.

The next step is to pass to the real form $\fsu(n)$ identified with the space of traceless Hermitian matrices of order $n$ with the Lie bracket given by~\eqref{eq:Lie_bracket}. To do this we introduce
\begin{align*}
    X_\alpha &:= Q_\alpha + Q_{-\alpha}\\
    Y_\alpha &:= i(Q_{\alpha}-Q_{-\alpha})
\end{align*}
and we obtain
$$
[H,X_\alpha] = -\alpha(H)Y_\alpha, \qquad [H,Y_\alpha] = \alpha(H)X_\alpha.
$$
meaning that $\fsu(n)$ under the adjoint action of $\fh$ can be effectively decomposed into $\fh$ and subspaces
$$
V_\alpha = \spann_\R\{X_\alpha,Y_\alpha\},
$$
where $\alpha \in \Delta$, and we see that $\fh$ acts as infinitesimal rotations inside $V_\alpha$.

We can now pass to the action of the original Lie group $G$. To do this, we restrict the adjoint action of $\fh$ on $\fsu(n)$ to $\fg$. Then if $\alpha(H) = 0$ for all $H\in \fg$, the action on $V_\alpha$ becomes trivial. Otherwise, it is still a rotation inside $V_\alpha$. It might also happen that for two positive roots $\alpha_1 \neq \alpha_2$ one has 
\begin{equation}
    \label{eq:root_eqivalent}
    \alpha_1|_{\fg} = \pm\alpha_2|_{\fg}.
\end{equation} 
In this case $V_{\alpha_1}\simeq V_{\alpha_2}$ as representations, and the action of $G$ will be the diagonal action on the sum of $\spann\{X_{\alpha_1},Y_{\alpha_1}\}$ and $\spann\{X_{\alpha_2},Y_{\alpha_2}\}$. Note that in order to get equivalent representations when the signs of $\alpha_1$ and $\alpha_2$ do not match, one needs to replace $Y_\alpha$ with $-Y_\alpha$. If two roots $\alpha_1,\alpha_2$ satisfy~\eqref{eq:root_eqivalent}, then we say that they are $\fg$-equivalent and denote it by $\alpha_1 \sim_\fg \alpha_2$.

Therefore, we have proven the following theorem.

 \begin{theorem}
     Let $G\subset SU(n)$ be a closed connected abelian subgroup, let $\fh$ be a maximal abelian subalgebra that contains $\fg$ and let $\Delta$ be the corresponding set of positive roots. Let $\tilde \Delta := \Delta/\sim_\fg$. Then $\fsu(n)$ is decomposed into subrepresentations of $G$ as
     $$
     \fsu(n) = \fh \oplus (\bigoplus_{\gamma \in \tilde \Delta} W_\gamma)
     $$
     where 
     $$
     W_\gamma = \bigoplus_{[\alpha]=\gamma}V_\alpha.
     $$
 \end{theorem}

In the next subsection we will use this theorem to prove the first part of Theorem~\ref{thm:qubit_networks}.

\subsection{Qubit networks with $X$-controls}
\label{sec:qubit_networks_X}
Consider the setting of the first part of Theorem~\ref{thm:qubit_networks}. In this case clearly $H_1,\dots,H_N$ generate an abelian Lie algebra given by $\spann\{X_1,\dots,X_N\}$. The first step is to understand the closure of the algebraic Lie group that integrates this Lie algebra. 

We start by changing the basis, so that the individual $X$ Hamiltonians become $Z$ Hamiltonians. This new basis will simplify the computation of restricted roots and root spaces. This can be achieved by applying $H^{\otimes N}$, where $H$, again by an abuse of notation, is the standard Hadamard matrix
$$
H = \frac{1}{\sqrt{2}}\begin{pmatrix}
1 & 1\\
1 & -1
\end{pmatrix}.
$$
It is straightforward to check that
$$
HXH^\dagger = Z, \quad HYH^\dagger = -Y, \quad HZH^\dagger = X.
$$

Therefore, the Hamiltonians from Theorem~\ref{thm:qubit_networks} in the new basis will be written as
\begin{align}
      \label{eq:supra_model_new}
    H_0 &= \omega \sum_{j=1}^N X_j + \sum_{j \sim k} \beta Z_jZ_k + \gamma Y_jY_k + \delta X_jX_k, \\
    H_j &= Z_j, j=1,\dots N  
\end{align}

The next step is to consider a maximal abelian subalgebra $\fh$ that contains all $Z_j$. Since we are considering $\fsu(2^N)$, the dimension of the Cartan subalgebra must be $2^N-1$. It is straightforward to see that all possible tensor products of $Z_j$ give rise to a Cartan subalgebra. 

After that we need to identify the corresponding subspaces $V_\alpha$. To do this we introduce the usual notation from quantum information:
$$
x,y,z\in \{0,1\}^N,
$$
and a slightly less usual one
$$
Z_{\ket x} = M_1\otimes \dots \otimes M_N,
$$
where
$$
M_j = \begin{cases}
I, & \text{if } x_j = 0,\\
Z, & \text{if } x_j = 1.
\end{cases}
$$
Then we have, by construction,
$$
Z_{\ket x}\ket{y} = (-1)^{x\cdot y}\ket y.
$$
We can order $\{0,1\}^N$ using lexicographic order. Then for $\ket y < \ket z$ we can define 
\begin{align*}
    X_{y,z} &= \ket y \bra z+ \ket z \bra y, \\
    Y_{y,z} &= i(\ket y \bra z- \ket z \bra y).
\end{align*}
A straightforward calculation then shows that
\begin{align*}
    [Z_{\ket x},X_{y,z}] = -\left((-1)^{x\cdot y}-(-1)^{x\cdot z}\right) Y_{y,z}\\
    [Z_{\ket x},Y_{y,z}] = \left((-1)^{x\cdot y}-(-1)^{x\cdot z}\right) X_{y,z}
\end{align*}
This is exactly the root space decomposition of $\fsu(2^N)$ relative to the Cartan subalgebra generated by $Z_{\ket x}$. The positive roots correspond to pairs $\ket y <\ket z$ and their evaluation on $Z_{\ket x}$ is given by
$$
(\ket y,\ket z)(Z_{\ket x}) = (-1)^{x\cdot y}-(-1)^{x\cdot z},
$$
and we can see that it can take values $\{-2,0,2\}$.

Let us now consider the restriction of $Z_j$, $j=1,\dots N$, which in our new notation is just $Z_{\ket{1_j}}$,  to each individual root space. We have that 
\begin{equation}
    \label{eq:root_space_values}
    (\ket y,\ket z)(Z_j) = \begin{cases}
    0, &\text{if } y_j= z_j;\\
    2, &\text{if } y_j=0 \text{ and } z_j= 1;\\
    -2, &\text{if } y_j=1 \text{ and } z_j= 0.
\end{cases}
\end{equation}
So we see that there is no root that could annihilate all of the $Z_j$. Indeed we have $\ket y \neq \ket z$, because if not, it would belong to the Cartan subalgebra $\fh$. Hence, there must be at least one digit where they are different. Therefore, the trivial subrepresentation of $\fg$ is exactly the Cartan subalgebra $\fh$.

Let us now identify the root spaces which have the same set of root values with respect to the elements $Z_j$. From~\eqref{eq:root_space_values} one can notice that one ambiguity in the value of $(\ket y,\ket z)(Z_j)$ can come when $y_j = z_j$ (they can be either both zero or both one). Therefore, we will have the same root values for two different pairs $(\ket y,\ket z)$ and $(\ket {\tilde y},\ket {\tilde z})$ if for some indices $j$ we have $y_j = z_j$ and $\tilde y_j = \tilde z_j$ and at the complementary indices $k$ we have $y_k = \tilde y_k$ and $z_k = \tilde z_k$. Another possibility could have been that the value of the root has the opposite sign. But this would require flipping to the opposite value both $y_k$ and $z_k$. But such an operation would reverse the order of $\ket y$ and $\ket z$. Therefore, only having $y_j = z_j$ and $\tilde y_j = \tilde z_j$ can produce root spaces with equivalent representations.

Now we pass to studying the orbitope $\conv \Ad_G H_0$. As discussed previously, the trivial subrepresentation is given exactly by the Cartan subalgebra $\fh$. Therefore, the projection to the subspace of fixed elements is given by
$$
\pi_0 H_0 = \beta\sum_{j\sim k} Z_jZ_k.
$$
Therefore, there is no way of decoupling individual qubits for the original Hamiltonians~\eqref{eq:qubit_hamiltonian}-\eqref{eq:control_hamiltonians} unless $\beta = 0$. 

We continue by determining projections of $\conv (Ad_G H_0)$ to the isotypic components. The Hamiltonians $X_j$ and $\gamma Y_jY_k + \delta X_jX_k$ each belong to different isotypic components. Indeed, 
\begin{itemize}
    \item Since $X_j$ flips only a single digit, it can be written as a sum of operators $X_{x,y}$, where $x_j=0$, $y_j = 1$ and for all the other indices $x_k=y_k$;
    \item Similarly, each $\gamma Y_jY_k + \delta X_jX_k$ can be written as a sum of operators $X_{x,y}$, where $x_s=y_s$ for $s\notin\{j,k\}$ and $x_j = x_k = 0 $ and $y_j = y_k = 1$, or $x_j = y_k = 0$ and $x_k = y_j = 1$. More precisely, we have
\begin{align*}
        \gamma Y_jY_k + \delta X_jX_k &= (\delta-\gamma)(\ket{0_j0_k}\bra{1_j1_k}+\ket{1_j1_k}\bra{0_j0_k} )\otimes I_{\Z_N\setminus \{j,k\}} + \\
        &+(\delta+\gamma)(\ket{0_j1_k}\bra{1_j0_k}+\ket{1_j0_k}\bra{0_j1_k} )\otimes I_{\Z_N\setminus \{j,k\}}
\end{align*}
\end{itemize}  
Therefore, each $X_j$ lies in a unique isotypic component, while $\gamma Y_jY_k + \delta X_jX_k$ lies in a sum of two isotypic components. To simplify the notation for the final result of this analysis, we denote
\begin{align*}
    X_{0_j,0_k} &:= (\ket{0_j0_k}\bra{1_j1_k}+\ket{1_j1_k}\bra{0_j0_k} )\otimes I_{\Z_N\setminus \{j,k\}} \\
    Y_{0_j,0_k} &:= i(\ket{0_j0_k}\bra{1_j1_k}-\ket{1_j1_k}\bra{0_j0_k} )\otimes I_{\Z_N\setminus \{j,k\}} \\
    X_{0_j,1_k} &:= (\ket{0_j1_k}\bra{1_j0_k}+\ket{1_j0_k}\bra{0_j1_k})\otimes I_{\Z_N\setminus \{j,k\}} \\
    Y_{0_j,1_k} &:= i(\ket{0_j1_k}\bra{1_j0_k}-\ket{1_j0_k}\bra{0_j1_k})\otimes I_{\Z_N\setminus \{j,k\}}
\end{align*}
Then, using that the action on each root space is a 2D rotation, we obtain
\begin{align*}
    \Aff (\Ad_G H_0) &= \beta\sum_{j\sim k} Z_jZ_k +\omega\left(\bigoplus_{j}\spann \{ X_j,Y_j\}\right) \oplus \\
    &\oplus (\delta - \gamma)\left(\bigoplus_{j\sim k}\spann \{X_{0_j,0_k}, Y_{0_j,0_k}\}\right) \oplus \\
    &\oplus (\delta + \gamma)\left(\bigoplus_{j\sim k}\spann \{X_{0_j,1_k}, Y_{0_j,1_k}\}\right).
\end{align*}

As a final step it remains to switch back to the original basis by conjugating by the Hadamard matrix. It is straightforward to check that
\begin{align*}
    HX_{0_j,0_k}H^\dagger &= \frac{1}{2}(Z_jZ_k - Y_jY_k),\\
    HY_{0_j,0_k}H^\dagger &= \frac{1}{2}(Y_jZ_k + Z_jY_k),\\
    HX_{0_j,1_k}H^\dagger &= \frac{1}{2}(Z_jZ_k + Y_jY_k),\\
    HY_{0_j,1_k}H^\dagger &= \frac{1}{2}(Y_jZ_k - Z_jY_k).
\end{align*}
This gives the final result~\eqref{eq:orbitope_many_qubits_X_pulses}.

\subsection{Qubit networks with $X,Y$ controls}
\label{sec:qubit_networks_XY}

In the previous example we saw that only applying $X$ pulses is not enough to effectively control the system, since the term $\beta\sum X_jX_k$ cannot be compensated for. What happens if we also add $Y$ pulses? In this case $\Lie\{X,Y\}\simeq \fsu(2)$, and the pulse group becomes a combination of copies of $SU(2)$ acting on individual qubits.
$$
G = \prod_{j=1}^N (SU(2)_j \otimes I_{\Z_N\setminus \{j\}})
$$
We now need to decompose $\fsu(2^N)$ with respect to irreducible representations and isotypic components. In this case, it can be seen directly that the spaces
$$
V_{j_1,\dots,j_k} := \fsu(2)_{j_1} \otimes \dots \otimes \fsu(2)_{j_k} \otimes I_{\Z_N\setminus \{j_1,\dots,j_k\}}
$$
form invariant subspaces and have no non-trivial invariant subspaces. Also, none of them are equivalent to each other. Therefore $V_{j_1,\dots,j_k}$ are also isotypic components and by Lemma~\ref{lemm:decomposition} if the projection of $H_0$ to one of those subspaces is not zero, then we can generate any Hamiltonian in a sufficiently small neighborhood of $0\in V_{j_1,\dots,j_k}$. In the case of the drift Hamiltonian~\eqref{eq:qubit_drift}, we see that one can generate Hamiltonians lying inside $V_{j,k}$ for $j\sim k$ and inside $V_j$, which are also covered via pulses.

\section*{Appendix A: Proof of Lemma~\ref{lem:0_inside} and Theorem~\ref{thm:quantitative_qubit_networks}}

We start proof of Lemma~\ref{lem:0_inside} by checking that $0$ belongs to the interior of an orbitope of a representation that has no nonzero fixed elements. It requires two ingredients: a proposition about representability of points inside a convex set as barycenters of measures, and Weil's integration formula.

\begin{definition}
    Let $\Omega$ be a non-empty subset of a locally convex space $X$ and let $\mu$ be a Radon measure on $\Omega$. A point $x\in X$ is said to be represented by $\mu$ if for all linear functions $f\in X^*$
    $$
    f(x) = \int_\Omega f(y) d\mu(y)
    $$
\end{definition}

We have the following classical proposition, which can be found in~\cite[Proposition 1.1]{Choquet}.
\begin{proposition}
\label{prop:choquet}
Let $\mathcal{K}$ be a compact subset of a locally convex topological vector space $X$ such that $\Omega =\operatorname{conv}(\mathcal{\overline{K}})$ is compact. Then for each Radon probability measure $\mu$, there exists a unique $x \in \Omega$ such that $x$ is represented by $\mu$.
\end{proposition}

\begin{theorem}[Weil's Integration Formula, cf.~{\cite[Theorem~2.51]{weyl}}]
Let $G$ be a compact, unimodular group and $H$ be a closed subgroup of $G$. Let $d\mu_G$ and $d\mu_H$ be left Haar measures on $G$ and $H$, respectively. Then there exists a regular Borel measure $d\mu_{G/H}$ on the quotient space $G/H$ which is invariant under the action of $G$ by left translation, such that for any $f \in C(G)$, we have:
\[
\frac{1}{|G|}\int_G f(g)\,dg
=
\frac{1}{|G/H|}\int_{G/H}\left(
\frac{1}{|H|}\int_{H}
f(gh) d\mu_{H}(h)\right)
d\mu_{G/H}(gH)
\]
where $g H$ represents the left coset in the quotient space $G/H$.
\end{theorem}

\begin{proof}[Proof of Lemma~\ref{lem:0_inside}]
    Fix $v\neq 0$. Then the integral of $\rho(g)v$ over $G$ is the projection to the space of fixed elements of the action. Under the assumptions, this means that
    $$
    \frac{1}{|G|}\int_G \rho(g)(v) d\mu_G(g) = 0
    $$
    Consider next the isotropy group $G_0\subset G$ of all group elements that fix $v$. Then $G/G_0$ is diffeomorphic to the orbit $O_v$ and carries a measure $\mu$, which is the pull-back measure of $\mu_{G/G_0}$. Hence, for any $f\in V^*$, Weil's formula gives
\begin{align*}
        0 = f(0) &= \frac{1}{|G|}\int_G f(\rho(g)v) d\mu_G(g) = \frac{1}{|G/G_0|}\int_{G/G_0}\left(
    \frac{1}{|G_0|}\int_{G_0}
    f(\rho(gh)v) d\mu_{G_0}(h)\right)
    d\mu_{G/G_0}(gG_0) = \\
    &= \frac{1}{|G/G_0|}\int_{G/G_0}f\left(\rho(g)
    \frac{1}{|G_0|}\int_{G_0}
    \rho(h)v d\mu_{G_0}(h)\right)\,
    d\mu_{G/G_0}(gG_0) = \\
    &= \frac{1}{|G/G_0|}\int_{G/G_0}f(\rho(g)v)
    d\mu_{G/G_0}(gG_0) = \frac{1}{|O_v|}\int_{O_v}f(x)
    d\mu(x) ,
\end{align*}
where in the second line we used the fact that $\rho$ is a representation, and in the third that $\rho(h)v = v$. All the assumptions of Proposition~\ref{prop:choquet} are satisfied. Hence, we have that $0$ is representable and therefore belongs to $\conv(O_v)$.

    The next part is to prove that zero is actually inside the interior of the orbitope. To do so we can use support functions of convex sets. It is known that given a compact convex set $K\subset \R^n$ with non-empty interior we have for $x\in K$ that $x\in \p K$ if and only if there exists $a\in \R^n \setminus\{0\}$ such that
    $$
    \langle a,x\rangle = \sup_{y\in K }\langle a,y\rangle,
    $$
    where $\langle\cdot,\cdot\rangle$ is a scalar product on $\R^n$. 

    Denote by 
    $$
    C_v:=\spann\{\rho(g)v:g\in G \}
    $$
    the linear space spanned by the orbitope $\conv (O_v)$. To prove that $0\in \operatorname{int} \conv(O_v)$, we need to show that for any $a\in C_v\setminus \{0\}$, we can find $y\in O_v$ such that
    $$
    \langle a,y \rangle <0.
    $$
    To do this, we integrate $\langle a,\rho(g) v \rangle$ over $G$. By the previous result we have that
    $$
    \frac{1}{|G|}\int_G \langle a, \rho(g) v\rangle d\mu(g) = 0.
    $$
    But then there are two options. Either $\langle a, \rho(g) v\rangle \equiv 0$, or it must take both positive and negative values as a function of $g$. 
    The first option is impossible, because $\rho(g)v$ spans the whole of $C_v$, so it cannot be contained in a strict subspace. Therefore, there must exist $g\in G$ such that $\langle a,\rho(g)v \rangle < 0$. Therefore, $0$ belongs to the interior of $\conv(O_v)$.

\end{proof}

We will now prove Theorem~\ref{thm:quantitative_qubit_networks}. The proof is based on two lemmas.

\begin{lemma}
\label{lem:convexification_matryoshka}
If $H\in \conv (\Ad_G H_0)$, then 
$$
\conv (\Ad_G H) \subset \conv (\Ad_G H_0)
$$
\end{lemma}
The proof is a direct application of definitions. Averaging along subgroups in combination with this lemma will allow us to identify certain sufficiently simple subsets of $\conv (\Ad_G \pi_0^\perp H_0)$ that will be convexified to produce a sufficiently big convex set $K$ of full dimension inside $\conv (\Ad_G \pi_0^\perp H_0)$. A similar idea was used in~\cite{Agrachev_chambrion}, where the authors use projections to invariant subspaces inductively instead of averaging over subgroup actions and construct a full-dimensional polytope inside the orbitope of effective Hamiltonians.

The second lemma is the following one.
\begin{lemma}
\label{lem:convexification}
    Let $\R^{n} = V_1\oplus \dots \oplus V_m$. Assume $\dim V_j\geq1$ and $r_j>0$ for every $j$. Consider spheres inside each $V_j$
    $$
    S_j = \{x_j\in V_j:\|x_j\|= r_j\}.
    $$
    Then
\begin{equation}
    \label{eq:convex_set_desecription}
        \conv\{0,S_j:j=1,\dots, m\}= \left\{x\in \R^{n}:\sum_{j=1}^m\frac{\|x_j\|}{r_j} \leq 1\right\}
\end{equation}
\end{lemma}

\begin{proof}
    It is clear that the set on the right of~\eqref{eq:convex_set_desecription} is convex and contains each $S_j$ for $j=1,\dots,m$ and zero. Therefore, the inclusion from left to right follows. Let us now see that the other inclusion holds as well. Assume $x$ belongs to the set on the right. Then we can write it as

$$
x = \lambda_0 \cdot 0 + \sum_{j=1}^m\lambda_j y_j,
$$
    where 
    $$
    y_j = \frac{r_jx_j}{\|x_j\|}\quad\text{if }x_j\neq0.
    $$
    If $x_j=0$, choose any $y_j\in S_j$. In all cases, set
$$
\lambda_j = \frac{\|x_j\|}{r_j}, \qquad \lambda_0 = 1 -\sum_{j=1}^m \lambda_j.
$$
\end{proof}

We will apply this lemma in both cases of Theorem~\ref{thm:quantitative_qubit_networks}. In the first case, we work with $\pi_0^\perp H_0$ which is orthogonal to $\fg$ and restore the fixed term $\pi_0H_0$ afterwards. In the second case, we work with $\pi_{\fg^\perp}H_0$ and add the freely available local directions in $\fg$.

Let us start with the first case. By exponentiating the action of $\fg$, we can write explicitly  $Ad_G (\pi_0^\perp H_0)$. It contains all elements of the form

\begin{align}
    &\omega\sum_{j}(\cos 2\theta_j Z_j -\sin 2\theta_j Y_j) + \\
    + &\frac{(\delta - \gamma)}{2}\sum_{j\sim k}\bigl[\cos 2(\theta_j + \theta_k)(Z_jZ_k-Y_jY_k)- \sin 2(\theta_j + \theta_k)(Y_jZ_k+Z_jY_k)\bigr] + \\
    + &\frac{(\delta + \gamma)}{2}\sum_{j\sim k}\bigl[\cos 2(\theta_j - \theta_k)(Z_jZ_k+Y_jY_k)- \sin 2(\theta_j - \theta_k)(Y_jZ_k-Z_jY_k)\bigr].\label{eq:orbit_qubit_network_X}
\end{align}
where $\theta_j \in [0,\pi)$. Here we clearly see that
$$
G \simeq T^N = \{(\theta_1,\dots,\theta_N)\in [0,\pi)^N\}.
$$

We start by isolating $Z_j$. To do this, we look for the biggest subgroup $\tilde G\subset G$ which fixes $Z_j$. This must be the subgroup with $\theta_j = 0$. Thus we average over $\tilde G = \{(\theta_1,\dots,\theta_N):\theta_j = 0\}$. It is straightforward to check that the projection of $\pi_0^\perp H_0$ to the space of invariant elements of $\tilde G$ is exactly $\omega Z_j$. Therefore, by Lemma~\ref{lem:convexification_matryoshka} we have
$$
\omega \{ \cos 2\theta Z_j - \sin 2\theta Y_j : \theta \in [0,\pi) \} = \Ad_G(\omega Z_j) \subseteq \conv(\Ad_G(\pi_0^\perp H_0)).
$$

Next we would like to decouple the remaining terms. Again we average $\pi_0^\perp H_0$ over the subgroup of $G$ given by $\theta_j = \theta_k = 0$. This leaves us with  
\begin{align}
    &\omega Z_j + \omega Z_k + \frac{(\delta - \gamma)}{2}(Z_jZ_k-Y_jY_k) +\frac{(\delta + \gamma)}{2}(Z_jZ_k+Y_jY_k).\nonumber
\end{align}
Then we can average this result over some discrete subgroups. We consider subgroups for which $\theta_l = 0$ for $l\notin\{j,k\}$. First, we apply the group generated by $(\theta_j = \theta_k = \pi/2)$. The action of this group reverses the signs in front of $Z_j$ and $Z_k$. After that we can apply the averaging over the subgroup generated by $(\theta_j = \theta_k = \pi/4)$, which isolates the term $Z_jZ_k+Y_jY_k$, or over the subgroup generated by $(\theta_j = -\theta_k = \pi/4)$, which isolates the term $Z_jZ_k-Y_jY_k$. Therefore, we find that also
\begin{align*}
\frac{\delta - \gamma}{2}\{\cos 2\theta(Z_jZ_k-Y_jY_k)+ \sin 2\theta(Y_jZ_k+Z_jY_k):\theta \in [0,\pi)\} \subseteq \conv(\Ad_G(\pi_0^\perp H_0)),\\
\frac{\delta + \gamma}{2}\{\cos 2\theta(Z_jZ_k+Y_jY_k)+ \sin 2\theta(Y_jZ_k-Z_jY_k):\theta \in [0,\pi)\} \subseteq \conv(\Ad_G(\pi_0^\perp H_0)).
\end{align*}

Summarizing everything, we see that $\conv(Ad_G(\pi_0^\perp H_0))$ must contain the convexification of a number of circles described above. Thus we can apply Lemma~\ref{lem:convexification}. It gives a sufficient condition for an effective Hamiltonian to live inside the orbitope. We have proven in Theorem~\ref{thm:qubit_networks} that $\conv (Ad_G \pi_0^\perp H_0)$ includes Hamiltonians of the form 
\begin{align}
    &\omega\sum_{j=1}^N(a_{j} Z_j +b_j Y_j) + \nonumber\\
    + &\frac{(\delta - \gamma)}{2}\sum_{j\sim k}\left(a_{jk}(Z_jZ_k-Y_jY_k)+ b_{jk}(Y_jZ_k+Z_jY_k)\right) + \nonumber\\
    + &\frac{(\delta + \gamma)}{2}\sum_{j\sim k}\left(\tilde a_{jk}(Z_jZ_k+Y_jY_k)+ \tilde b_{jk}(Y_jZ_k-Z_jY_k)\right). \nonumber
\end{align}
By Lemma~\ref{lem:convexification}, the following condition on the coefficients is sufficient:
\begin{align}
    &\sum_{j=1}^N\frac{\|(a_{j} Z_j +b_j Y_j)\|}{\|Z_j\|} + \sum_{j\sim k}\frac{\|a_{jk}(Z_jZ_k-Y_jY_k)+ b_{jk}(Y_jZ_k+Z_jY_k)\|}{\|Z_jZ_k-Y_jY_k\|} + \nonumber\\
    + &\sum_{j\sim k}\frac{\|\tilde a_{jk}(Z_jZ_k+Y_jY_k)+ \tilde b_{jk}(Y_jZ_k-Z_jY_k)\|}{\|Z_jZ_k+Y_jY_k\|} \leq 1,\label{eq:ineq_effective_hamiltonians}
\end{align}
For nonzero amplitudes, we have canceled the absolute values of $\omega$, $(\delta-\gamma)/2$ and $(\delta+\gamma)/2$ that appear simultaneously in the numerator and denominator. Components with zero amplitude are omitted when applying Lemma~\ref{lem:convexification}. Thus it remains to compute individual norms. Using the fact that
$\Tr(A\otimes B) = (\Tr A)(\Tr B)$, we see that $Z_j$ is orthogonal to $Y_j$ and $\{Z_jZ_k,Y_jY_k,Y_jZ_k,Z_jY_k\}$ form an orthogonal set. In addition, we have
$$
\|Z_j\|^2 = \Tr(Z_j^2) = \Tr(I^{\otimes N})= \Tr(I)^N = 2^N
$$
Similarly, we find $\|Y_j\| = 2^{\frac{N}{2}}$ and
$$
\|Z_jZ_k-Y_jY_k\|=\|Z_jZ_k+Y_jY_k\|=\|Y_jZ_k-Z_jY_k\|=\|Y_jZ_k+Z_jY_k\|=2^{\frac{N+1}{2}}.
$$
Thus the inequality~\eqref{eq:ineq_effective_hamiltonians} reduces to
$$
\sum_{j=1}^N \sqrt{a_j^2+b_j^2} + \sum_{j\sim k}\left( \sqrt{a_{jk}^2+b_{jk}^2} + \sqrt{\tilde a_{jk}^2+\tilde b_{jk}^2}\right) \leq 1.
$$

Now we prove the second part when both $X_j$ and $Y_j$ lie inside $\fg$ for all $j = 1,\dots,N$. Let us consider the drift Hamiltonian~\eqref{eq:qubit_drift}. Then it will have trivial projections to $V_{j_1,\dots,j_k}$ when $k\geq 3$. Subspaces $V_{j}$ can be trivially generated using the pulse group $G$ and the projection of $H_0$ to $V_{j,k}$ is equal to $(\beta X_{j}X_{k}+\gamma Y_{j}Y_{k}+\delta Z_jZ_k)$, when $j\sim k$. 

 Here $V_0=\{0\}$, so $\pi_0^\perp H_0=H_0$. Since all Hamiltonians in $\fg$ can be generated, it only remains to focus on Hamiltonians inside $\fg^\perp$. They well be generated by
$$
H_0^\perp:=\pi_{\fg^\perp}H_0=H_0-\omega\sum_{j=1}^N Z_j=\sum_{j\sim k}\bigl(\beta X_jX_k+\gamma Y_jY_k+\delta Z_jZ_k\bigr).
$$
Since $\tilde\cH=\conv(\Ad_G H_0^\perp)+\fg$, we first bound $\conv(\Ad_G H_0^\perp)$ from inside. We average $H_0^\perp$ over the subgroup of $G$ given by
$$
 \prod_{l\notin\{j,k\}} (SU(2)_l \otimes I_{\Z_N\setminus \{l\}}),
$$
i.e. the maximal subgroup that acts trivially on the $j$-th and $k$-th qubits. Then, after the averaging we obtain the Hamiltonian
$$
\beta X_jX_k + \gamma Y_jY_k + \delta Z_jZ_k.
$$
Next, we would like to find the maximal-volume ellipsoid that lies inside $\conv Ad_G (\beta X_jX_k + \gamma Y_jY_k + \delta Z_jZ_k)$. We can replace $G$ with the action of two copies of $SU(2)$ on qubits $j$ and $k$. This representation is irreducible. It was proven in~\cite{Barvinok2005} that for irreducible representations the maximal ellipsoid is a sphere centered at zero. Let us compute this sphere.

First we note that the orbit of the adjoint action of $SU(2)$ on $\fsu(2)$ can be identified with the orbit of the standard $SO(3)$ action on $\R^3$. Therefore, the orbit of the action of $SU(2)\times SU(2)$ on $V_{j,k}$ can be identified with the orbit of $SO(3)\times SO(3)$ on $\R^3\otimes \R^3$. The space $\R^3\otimes \R^3$ can be further identified with the space of $3\times 3$ real matrices. It would be natural to identify, for example, $X_jX_k$ with the matrix $E_{11}$ (the one whose only non-zero element is $(E_{11})_{11} = 1$), however, this would not give an isometry between the two spaces. In order to make this map an isometry with $\mathfrak{gl}(3,\R)$ endowed with the Frobenius norm, we need to further multiply all of the elements by their length in $\fsu(n)$, which, as computed in the first part of the proof, is given by $2^{\frac{N}{2}}$. For example, we have the identifications
\begin{align*}
    X_jX_k &\simeq 2^{\frac{N}{2}} E_{11},\\
    X_jY_k &\simeq 2^{\frac{N}{2}} E_{12},\\
    X_jZ_k &\simeq 2^{\frac{N}{2}} E_{13},\\
    Y_jX_k &\simeq 2^{\frac{N}{2}} E_{21},\\
    &\dots
\end{align*}
and so on. The action of $SO(3)\times SO(3)$ becomes
\begin{equation}
    \label{eq:so(3)-action}
    A\mapsto R_1AR_2^T.
\end{equation}
This orbitope is well known in the literature. Here is an explicit description from~\cite{Miranda1994} (see the first theorem on page 138).

\begin{proposition}
    Given a matrix $A\in \mathfrak{gl}(3,\R)$ denote by $\sigma_1(A)\geq \sigma_2(A)\geq \sigma_3(A) \geq 0$ its singular values. Consider the action~\eqref{eq:so(3)-action}. Then the corresponding orbitope consists of all matrices $X\in\mathfrak{gl}(3,\R)$, whose singular values satisfy
    \begin{align*}
        \sigma_1(X) &\leq \sigma_1(A),\\
        \sigma_1(X)+\sigma_2(X)+\sign (\det X)\sigma_3(X) &\leq \sigma_1(A)+\sigma_2(A)+\sign (\det A)\sigma_3(A),\\
        \sigma_1(X)+\sigma_2(X)-\sign (\det X)\sigma_3(X) &\leq \sigma_1(A)+\sigma_2(A)-\sign (\det A)\sigma_3(A).
    \end{align*}
\end{proposition}

We see that this gives (essentially) three linear inequalities for the singular values. In addition, we know that the Frobenius norm of a matrix $X$ satisfies
$$
\|X\|_2 = \sqrt{\sum_{i=1}^3 \sigma_i(X)^2}.
$$
Therefore, the problem of finding the largest inscribed sphere corresponds to finding a sphere in $\R^3$ with coordinates $\sigma_i(X)$ that satisfies the three linear constraints above and in addition $\sigma_1(X)\geq \sigma_2(X)\geq \sigma_3(X) \geq 0$. But then, it means that we only need to find the smallest distance between the origin and each of the hyperplanes defined by replacing each inequality with equality. We get the following possibilities for the limiting radius $r$ depending on the sign of $\det X$.
\begin{enumerate}
    \item $\sign \det X > 0$:
    $$
    r = \min\left\{\sigma_1(A),\frac{\sigma_1(A)+\sigma_2(A)+\sign (\det A)\sigma_3(A)}{\sqrt{3}},\frac{\sigma_1(A)+\sigma_2(A)-\sign (\det A)\sigma_3(A)}{\sqrt 2}\right\}
    $$
    \item $\sign \det X < 0$:
    $$
    r = \min\left\{\sigma_1(A),\frac{\sigma_1(A)+\sigma_2(A)+\sign (\det A)\sigma_3(A)}{\sqrt{2}},\frac{\sigma_1(A)+\sigma_2(A)-\sign (\det A)\sigma_3(A)}{\sqrt{3}}\right\}
    $$
    \item $\sign \det X = 0$:
    $$
    r=\min\left\{\sigma_1(A),\frac{\sigma_1(A)+\sigma_2(A)+\sign (\det A)\sigma_3(A)}{\sqrt{2}},\frac{\sigma_1(A)+\sigma_2(A)-\sign (\det A)\sigma_3(A)}{\sqrt{2}}\right\}
    $$
\end{enumerate}

\begin{remark}
    The distances to the constraint boundaries follow from
$$
        \min_{\substack{x_1+x_2+x_3=c\\x_1 \geq x_2 \geq x_3\geq0}} \|x\| = \frac{c}{\sqrt{3}}, \qquad \min_{\substack{x_1+x_2-x_3=c\\x_1 \geq x_2 \geq x_3\geq0}} \|x\| = \frac{c}{\sqrt{2}},
$$
for $c\geq0$, which can be understood geometrically.
\end{remark}

Taking the global minimum among all choices, we see that we must choose the radius of the sphere as
$$
r_{min} = \min\left\{\sigma_1(A),\frac{\sigma_1(A)+\sigma_2(A)-\sigma_3(A)}{\sqrt{3}}\right\}.
$$

We can now apply Lemma~\ref{lem:convexification}. To do so, we identify
$$
\beta X_jX_k + \gamma Y_jY_k + \delta Z_jZ_k \simeq 2^{\frac{N}{2}}(\beta E_{11}+\gamma E_{22}+\delta E_{33}) = : M.
$$
Since this matrix is diagonal, its singular values are equal to $2^{\frac{N}{2}}|\beta|$, $2^{\frac{N}{2}}|\gamma|$, $2^{\frac{N}{2}}|\delta|$. In particular,
$$
\sigma_1(M) = 2^{\frac{N}{2}} \max\{|\beta|,|\gamma|,|\delta|\}
$$
and
$$
\sigma_1(M)+\sigma_2(M)-\sigma_3(M) = 2^{\frac{N}{2}}\max\{|\beta|+|\gamma|-|\delta|,|\gamma|+|\delta|-|\beta|,|\delta|+|\beta|-|\gamma|\}.
$$

The ball obtained for each edge has radius $2^{N/2}R$. For an edge Hamiltonian
$$
H_{jk} :=\sum_{p,q=1}^3a_{pq,jk}P_{p,j}P_{q,k},\qquad (P_1,P_2,P_3)=(X,Y,Z),
$$
we have
$$
\|H_{jk}\|=2^{N/2}\sqrt{\sum_{p,q=1}^3a_{pq,jk}^2}.
$$
Lemma~\ref{lem:convexification} therefore gives the sufficient condition
$$
\sum_{j\sim k}\frac{\|H_{jk}\|}{2^{N/2}R}
=\frac1R\sum_{j\sim k}\sqrt{\sum_{p,q=1}^3a_{pq,jk}^2}\leq1.
$$
Adding any element of $\fg$ supplies arbitrary local coefficients, including when $\omega=0$, and proves the second part of the theorem.

\newpage

\bibliographystyle{alpha}
\bibliography{references}

\end{document}